\documentclass[journal,comsoc]{IEEEtran}
\usepackage{amsmath,amsfonts}
\usepackage{algorithm}
\usepackage{array}
\usepackage{algpseudocode}
\usepackage[caption=false,font=normalsize,labelfont=sf,textfont=sf]{subfig}
\usepackage{textcomp}
\usepackage{stfloats}
\usepackage{url}
\usepackage{verbatim}
\usepackage{amsthm,amssymb}
\usepackage{graphicx}
\usepackage{color}
\usepackage{cite}
\newtheorem{theorem}{Theorem}
\begin{document}

\title{Attacking The Assortativity Coefficient Under A Rewiring Strategy}
\author{Shuo Zou, Bo Zhou, and~Qi~Xuan,~\IEEEmembership{Senior Member,~IEEE,}
\thanks{This work was supported in part by the National Natural Science Foundation of China under Grant 61973273, by the Zhejiang Provincial Natural Science Foundation of China under Grant LR19F030001, by the National Key R\&D Program of China under Grant 2020YFB1006104, and by the Research and Development Center of Transport Industry of New Generation of Artificial Intelligence Technology. \emph{(Corresponding authors: Qi Xuan.)}}
\thanks{
All authors are with the Institute of Cyberspace Security, College of Information Engineering, Zhejiang University of Technology, Hangzhou 310023, China.} 
}



\maketitle

\begin{abstract}
Degree correlation is an important characteristic of networks, which is usually quantified by the assortativity coefficient. However, concerns arise about changing the assortativity coefficient of a network when networks suffer from adversarial attacks. In this paper, we analyze the factors that affect the assortativity coefficient and study the optimization problem of maximizing or minimizing the assortativity coefficient ($r$) in rewired networks with $k$ pairs of edges. We propose a greedy algorithm and formulate the optimization problem using integer programming to obtain the optimal solution for this problem. Through experiments, we demonstrate the reasonableness and effectiveness of our proposed algorithm. For example, rewired edges 10\% in the ER network, the assortativity coefficient improved by 60\%.

\end{abstract}


\section{Introduction}
\IEEEPARstart{C}{omplex} networks serve as powerful tools for abstractly representing real-world systems, where individual units are represented as nodes, and interactions between these units are represented as edges.
Therefore, research on complex networks has experienced tremendous growth in recent years. Various network properties, including the degree sequence\cite{chung2002connected,chatterjee2011random}, degree correlation\cite{park2003origin,mahadevan2006systematic} and clustering coefficient\cite{saramaki2007generalizations,mcassey2015clustering} are extensively utilized in complex network analysis to assess the topological structure of networks. The degree correlation, one of the most intriguing properties in complex networks, describes the relationship between the degrees of connected nodes. This crucial topological characteristic plays an important role in network stability\cite{friedel2007influence}, attack robustness\cite{huang2011robustness}, network controllability\cite{zhang2019evolution}, information propagation\cite{vega2020influence} and other related phenomena\cite{noldus2013effect,zhou2014memetic,oh2018complex,osat2017optimal,olvera2021pagerank,li2021percolation}. The assortativity coefficient\cite{newman2002assortative,newman2003mixing} is a commonly used measure to quantify the degree correlation in complex networks.

There has been some research on degree correlation attacks, which can be broadly classified into two types: attacks that disrupt the degree sequence and attacks that preserve the degree sequence before and after the attack. Menche~\emph{et al.}\cite{menche2010asymptotic} analyzed in detail a class of maximally correlated scale-free networks. They found the asymptotic properties of degree-correlated scale-free networks. 
Srivastava~\emph{et al.}\cite{srivastava2012correlations} examined the effects of node or edge deletion on degree correlation and discovered that removing nodes or edges can transform an initially assortative network into a disassortative one, and vice versa. Xulvi~\emph{et al.}\cite{xulvi2005changing} proposed two algorithms that aim to achieve the desired degree correlation in a network by producing assortative and disassortative mixing, respectively.  Li~\emph{et al.}\cite{jing2016algorithm} developed a probabilistic attack method that increases the chances of rewiring the edges between nodes of higher degrees, leading to a network with a higher degree of assortativity. Geng~\emph{et al.}\cite{geng2021global} introduced a global disassortative rewiring strategy aimed at establishing connections between high-degree nodes and low-degree nodes through rewiring, resulting in a higher level of disassortativity within the network. However, the above work overlooks the fact that attackers typically intelligently obtain the initial state of the network structure. It is important to consider how to maximize or minimize the assortativity coefficient through a rewiring strategy based on the initial network structure.
 
In this paper, we investigate a greedy rewiring strategy designed to manipulate the assortativity coefficient of a network while preserving its degree sequence. We make the following contributions:
\begin{itemize}
    \item We define the assortativity coefficient attack and propose a rewiring strategy based on the original graph.
    \item We have shown that the objective function under the rewiring strategy is monotone and submodular.
    \item We propose a greedy rewiring strategy and demonstrate its effectiveness on 2 small real-world networks. Furthermore, we comprehensively evaluate our method by comparing it with $3$ baselines on $6$ datasets.
    
\end{itemize}


\section{Problem statement}\label{sec}
\subsection{Preliminaries}
We consider a graph $G=(V,E)$, where the set of vertex $V$ is a set of $N$ nodes, and $E$ is a set of undirected edges $M$. The assortativity coefficient is a widely used measure to quantify the degree correlation in a network. 
The assortativity coefficient is defined as\cite{newman2003mixing}:
\begin{equation}
    \mathbf{r} = \frac{M^{-1}\sum_i^M(j_{i}k_{i})-[M^{-1}\sum_i^M\frac{1}{2}(j_{i}+k_{i})]^2}{M^{-1}\sum_i^M\frac{1}{2}(j_{i}^2+k_{i}^2)-[M^{-1}\sum_i^M\frac{1}{2}(j_{i}+k_{i})]^2}.
    \label{eq:1}
\end{equation}
where $k_i$ and $j_i$ are the degrees of the endpoints of edge $i$, respectively. 

The degree distribution is a crucial characteristic of a network as it reveals the connectivity patterns and the overall topology of the network. Therefore, we propose an attack method that preserves the degree of the node.
The rewiring strategy is shown in Figure~\ref{fig:99}. We choose the edge pair $(i,j) \in E$ and $(k,l)\in E$, which can be rewired as $(i,k)$ and $(j,l)$ if $(i,k), (j,l) \notin E$, or can be rewired as $(i,l)$ and $ (k,j)$ if $(i,l),(k,j) \notin E$.
Obviously, the rewiring strategy does not change the degree of the nodes. $\sum_i^{M}\frac{1}{2}(j_i^2+k_i^2)$ and $\sum_i^{M}\frac{1}{2}(j_i+k_i)$ are also unchanged under the rewiring strategy. According to Formula~\ref{eq:1}, our rewiring strategy only affects the following formula:
\begin{equation}
    \mathbf{p} = \sum_i^M(j_ik_i).
    \label{eq:02}
\end{equation}

When the edges ($i$,$j$) and ($k$,$l$) are rewired to form ($i$,$k$) and ($j$,$l$), the change in the assortativity coefficient can be converted to the change in $p$, calculated as $value_{\{(i,j),(k,l)\}}=(k_i k_k + k_j k_l)-(k_i k_j + k_k k_l)$.
\begin{figure}[ht]
\centering
\includegraphics[scale=0.2]{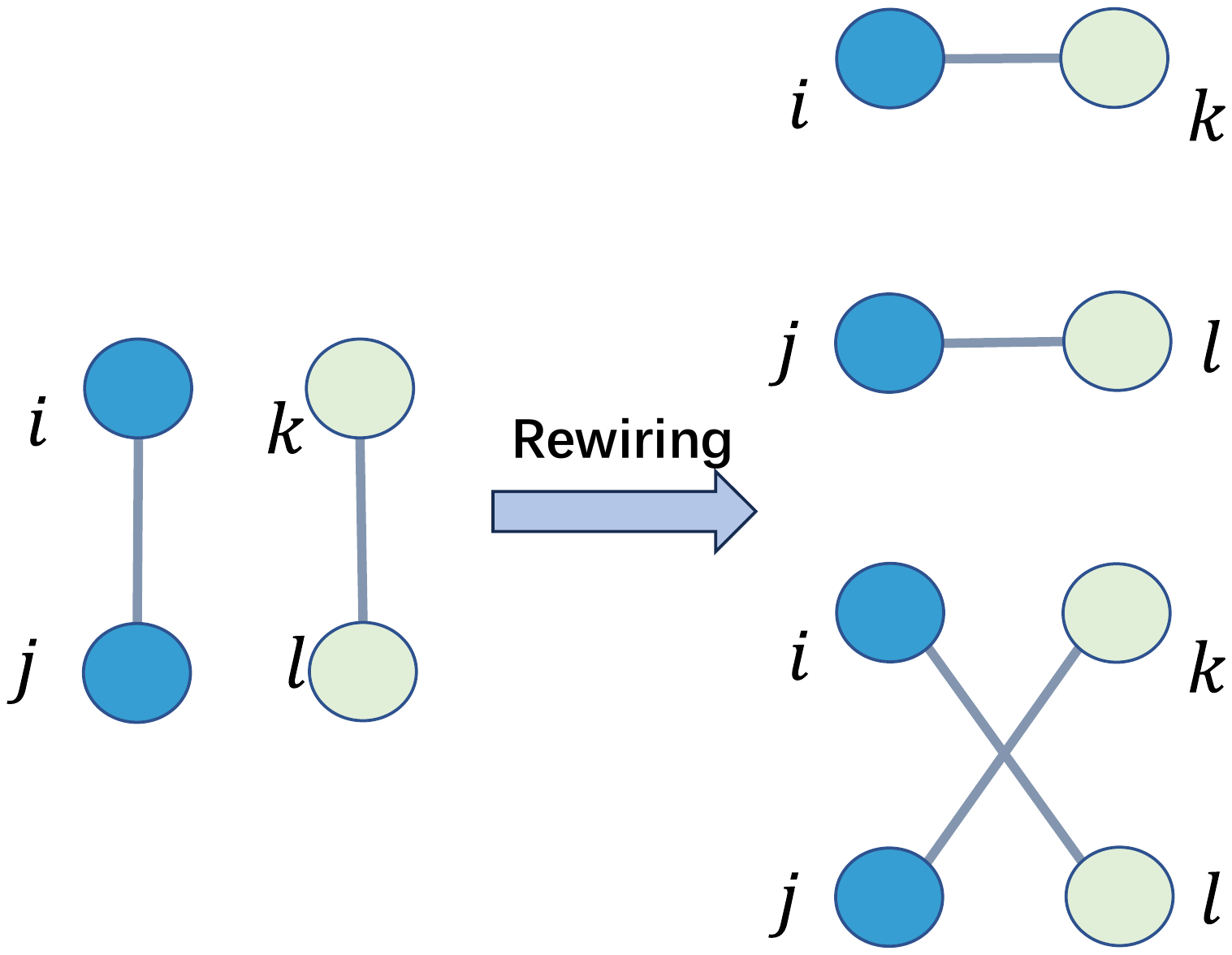}
\caption{The degrees of nodes $i$, $j$, $k$, and $l$ are $4$, $1$, $3$, and $2$, respectively. The rewiring of the edge pairs $(i,j)$ and $(k,l)$ can occur in two possible ways, corresponding to $value_{\{(i,j),(k,l)\}}=(4 \times  3 + 1 \times 2)-(4 \times 1 + 3 \times 2)=4$ and $value_{\{(i,j),(l,k)\}}=(4 \times  2 + 1 \times 3)-(4 \times 1 + 3 \times 2)=1$.}
\label{fig:99}
\end{figure}
\subsection{Problem Definition}
For a simple network $G(V, E)$, let $S$ be the set of rewired edge pairs. We denote the network after rewiring as $G+S$. The assortativity coefficient of $G+S$ is represented by $r(S)$, and the change in the assortativity coefficient can be expressed as $\Delta r(S)$.

Here, we assume that newly generated edge pairs resulting from rewiring will not be considered for further rewiring in subsequent steps.
Therefore, we can identify all potential edge pairs within the original graph without considering the additional components during the rewiring process. 

An adversary aims to maximize the change in the assortativity coefficient through a limited number of rewirings, including \textbf{ Maximum Assortative Rewiring (MAR)} and \textbf{Maximum Dissortative Rewiring (MDR)}. We define the following set function optimization problem: 

\begin{equation}
    \underset{S \subset EP,|S|=k}{maximize} \quad |\Delta r(S)|.
    \label{eq:2}
\end{equation}

where $EP$ is a set of rewirable edges. Since the change in the assortativity coefficient can be converted to the change in $p$, the optimization problem (\ref{eq:2}) is equivalent to the following problem:
\begin{equation}
    \underset{S \subset EP,|S|=k}{maximize} \quad |\Delta p(S)|.
    \label{eq:3}
\end{equation}
In MAR, the set $EP$ consists of all possible rewired edge pairs with a positive $value$ in the original $G$. In MDR, the set $EP$ consists of all possible rewired edge pairs with negative $value$. These edge pairs in $EP$ satisfy two mutually exclusive conditions. First, the pair of edges formed by the same edge and other edges are mutually exclusive, as each edge can only be rewired once. Second, edge pairs that result in the same edge after rewiring are also mutually exclusive, since simple graphs do not allow multiple edges between the same pair of nodes. 
\begin{theorem}\label{theorem1}
In the MAR and MDR problems, $\Delta p(S)$, exhibits monotonic behavior.
\end{theorem}
\begin{proof}
In MAR, for any given solution $S$, let us consider a pair of rewired edge pairs ${(i, j), (k, l)}$ in $G+S$ that can be rewired. The change in the assortativity coefficient, denoted $\Delta p(S \cup \{(i, j), (k, l)\})$, can be expressed as $\Delta p(S \cup \{\{(i, j), (k, l)\}\}) = \Delta p(S) + value_{\{(i, j), (k, l)\}}$. Since $value_{\{(i, j), (k, l)\}}>0$, it follows that $\Delta p(S \cup \{(i, j), (k, l)\}) > \Delta p(S)$, indicating that $p(S)$ is increasing monotonically. Similarly, it can be shown that in MDR, $\Delta p(S)$ is monotonically decreasing.
\end{proof}
\begin{algorithm}[!t]
  \caption{GRS} 
  \begin{algorithmic}[1]
    \Require
      Graph $G=(V,E)$; an integer $k$
    \Ensure
       A set $S$ and $|S|=k$
    \If{GARS}
        \State $EP \leftarrow $ the set of possible rewired edge pairs with a positive $value$ in the original $G$, sorted in descending order.
    \EndIf
    \If{GDRS}
        \State $EP \leftarrow $ the set of possible pairs of rewired edges with a negative $value$ in the original $G$, sorted in ascending order.
    \EndIf
    \State $S \leftarrow \emptyset$
    \State $index \leftarrow 0$
    \State $n \leftarrow 0$
    \State $len \leftarrow length(rewiringEdgeList)$
    \While {$n < k$ and $index < len$}
        \State edge $(i,j),(k,l)  \leftarrow rewiringEdgeList [index]$
        \State $index \leftarrow index + 1$
        \If{the edges $(i,k)$ and $(j,l)$ can be rewired in $G$}
            \State $S \leftarrow S \cup \{\{(i,j),(k,l)\}\}$
            \State $G \leftarrow G+\{\{(i,j),(k,l)\}\}$
            \State $n \leftarrow n + 1$
        \EndIf
    \EndWhile
    \State
    \Return {$S$}
  \end{algorithmic}
  \label{alg:1}
\end{algorithm}
\begin{theorem}\label{theorem2}
In the MAR and MDR problems, $\Delta p(S)$ is submodular.
\end{theorem}

\begin{proof}
For each pair $S$ and $T$ of MAR such that $S \subseteq T$,and for each pair of rewired edge pairs ${(a, b), (c, d)}$ in $G(S)$ that satisfy the rewiring requirements, $p(S \cup \{(a, b),(c, d)\}) - p(S) = p(T \cup \{(a, b),(c, d)\})-p(T) = value_{(a,b),(c,d)}$, so $\Delta p(S)$ is submodular. This is also the case in MDR.
\end{proof}

\subsection{Attack Method}
Theorem~\ref{theorem1} and~\ref{theorem2} indicate that the objective function (\ref{eq:3}) is both monotone and submodular. As a result, a simple greedy strategy can be used to approximate the problem (\ref{eq:2}) with provable optimality bounds. We propose the \textbf{Greedy Rewiring Strategy(GRS)} to maximize or minimize the assortative coefficient.

\textbf{Greedy assortative rewiring strategy(GARS):}
First, identify all possible pairs of rewired edges with a positive $value$ in the original graph $G$. Initialize the set $S$ is empty. Then select the pair with the highest positive $value$ and try to rewire it. If successful, add it to $S$. if not, move on to the pair with the second highest $value$ and repeat the process until $|S|=k$.

\textbf{Greedy disassortative rewiring strategy(GDRS):}
Similarly, begin by finding all possible pairs of rewired edges with a negative $value$ in the original graph $G$. The set $S$ is empty. Next, choose the pair with the smallest negative $value$ and try to rewire it. If it can be rewired, add it to $S$. if not, proceed with the pair with the second-smallest $value$ and repeat the process until $|S|=k$.

The details of this algorithm are summarized in Algorithm \ref{alg:1}. In fact, the time complexity of the algorithm is $O(M^2\log(M))$, where $M$ represents the number of edges in the graph.

\renewcommand{\algorithmicrequire}{\textbf{Input:}}  
\renewcommand{\algorithmicensure}{\textbf{Output:}} 

\section{Experiments}\label{thre}
In this section, we first validate the quality of the solution and the rationality of the GRS algorithm. We then compare it with several baseline methods to demonstrate that the GRS achieves the greatest change in the network's assortativity coefficient. We conducted experiments on both synthetic and real-world networks.
\subsection{Datasets}
We conducted our experiments on three popular model networks, Erdős–Rényi(ER)\cite{erdHos1960evolution} network, Watts-Strogatz(WS)\cite{watts1998collective} and Barabási-Albert(BA)\cite{barabasi1999emergence} network, and some realistic networks from KONECT\cite{kunegis2013konect}. The information of these networks is shown in Table \ref{tab1}.
\begin{table}[ht]
\centering
\caption{Statistics of datasets}\label{tab1}
\begin{tabular*}{\hsize}{@{}@{\extracolsep{\fill}}llcccccr@{}}
\hline
Networks & Nodes & Edges  & $\langle k \rangle$ & $r$ 
\\ \hline
ER   & 1000 & 5000& 10 & -0.004
\\ \hline
WS   & 1000 & 5000& 10 & -0.039
\\ \hline
BA   & 1000 & 4975& 9.95 & -0.056
\\ \hline
Karate   & 34 & 78& 4.489 & -0.476
\\ \hline
Dolphin   & 62 & 159& 5.129 & -0.044
\\ \hline
Powergrid   & 4941 & 6594& 2.669 & 0.003
\\ \hline
Netscience   & 1461 & 2742 & 3.753& 0.462
\\ \hline
Metabolic   & 1039 & 4741 & 9.126 &-0.250
\\ \hline
\end{tabular*}
\end{table} 
\subsection{Evaluating the Solution Quality and the Algorithm Rationality}
To evaluate the solution quality and rationality of our algorithm, we compared the results obtained from our algorithm with those of the optimal solutions, an improved version of our algorithm called \textbf{RenewGreedy} (which updates the candidate edge set after each rewiring), and the maximum or minimum assortativity coefficient of the network while preserving the degree. The evaluation was carried out on two small real networks, the Karate network and the Dolphin network.  

We employ integer programming formulations for the MAR and MDR problems to compute optimal solutions. Specifically, in the MAR problem, it is defined as follows:
\begin{equation*}
\begin{split}
&\max \,\,\sum_{ep \in EP}{value_{ep}x_{ep}} \\
&s.t.\quad  \left\{\begin{array}{lc}
\sum_{\{ep\in EP 
| (a,b) \in ep\}}
x_{ep} \leq 1 \,\,\,for \,\,each \, \,(a,b) \in E\\
\sum_{\{ep\in EP | (a,b) \in ep_{r}\}}
x_{ep}\leq 1 \,\,\,for \,\,each\,\, (a,b) \in E_{r}\\
\sum_{ep \in EP} x_{ep} \leq k
\\
x_{ep} \in \{0,1\} \,\,\,for\,\,each\,\,ep \in EP \\
\end{array}\right.
\end{split}
\end{equation*}
$E_r$ is a set of new edges generated after rewiring the elements in $EP$, and $ep_r$ represents the set of new edges generated after rewiring $ep$.  The first constraint ensures that each edge in the original graph can only be rewired once. The second constraint ensures that each new edge is only generated once. Similarly, in the MDR problem, we aim to minimize the objective function.

\begin{figure} [ht]
	\centering
	\subfloat[Karate]{
		\includegraphics[width=0.4\linewidth]{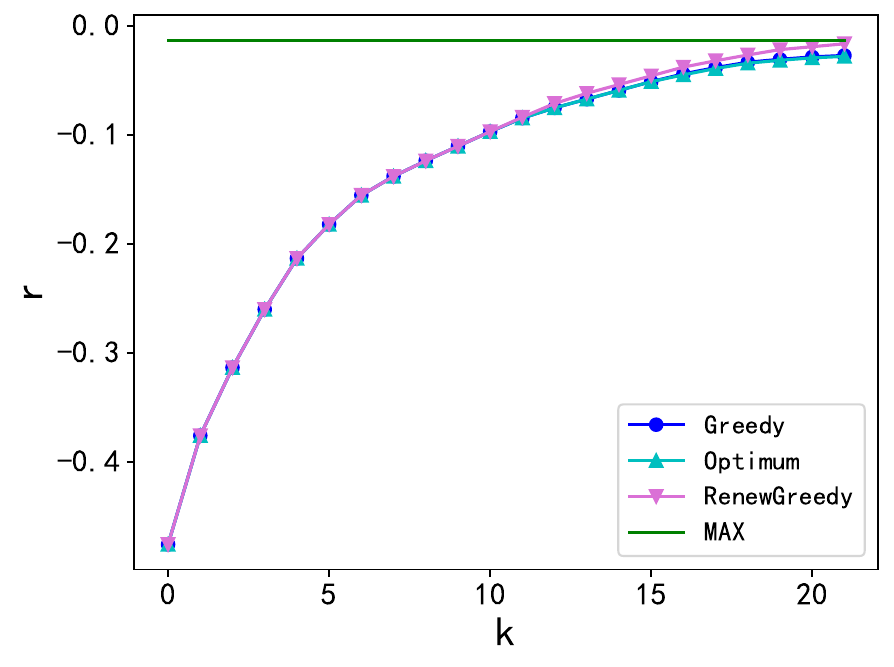}}
  \hspace{0.5cm}
	\subfloat[Karate]{
		\includegraphics[width=0.4\linewidth]{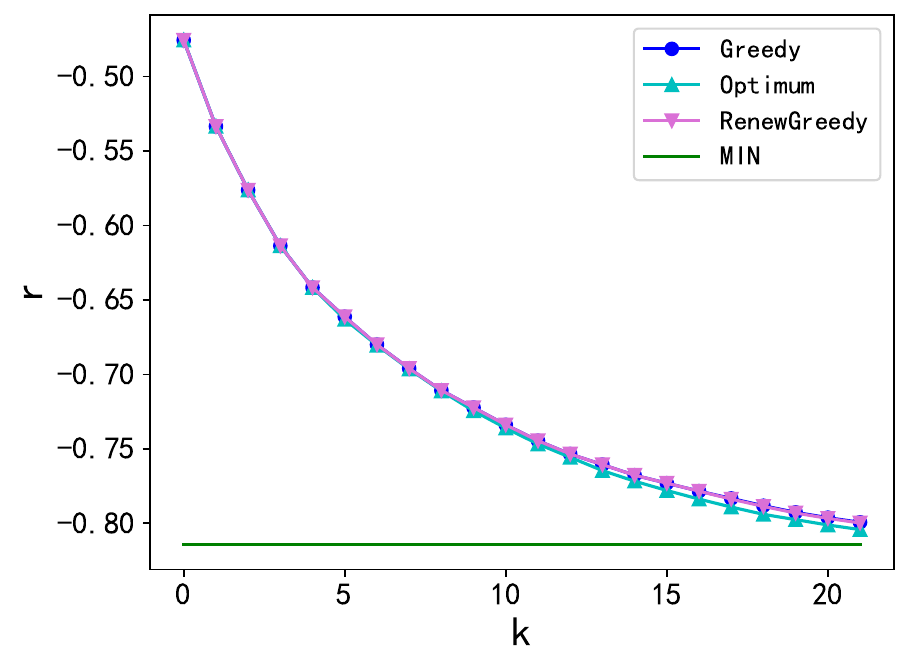} }
  \vspace{-0.3cm}
	\\
         \subfloat[Dolphin]{
		\includegraphics[width=0.4\linewidth]{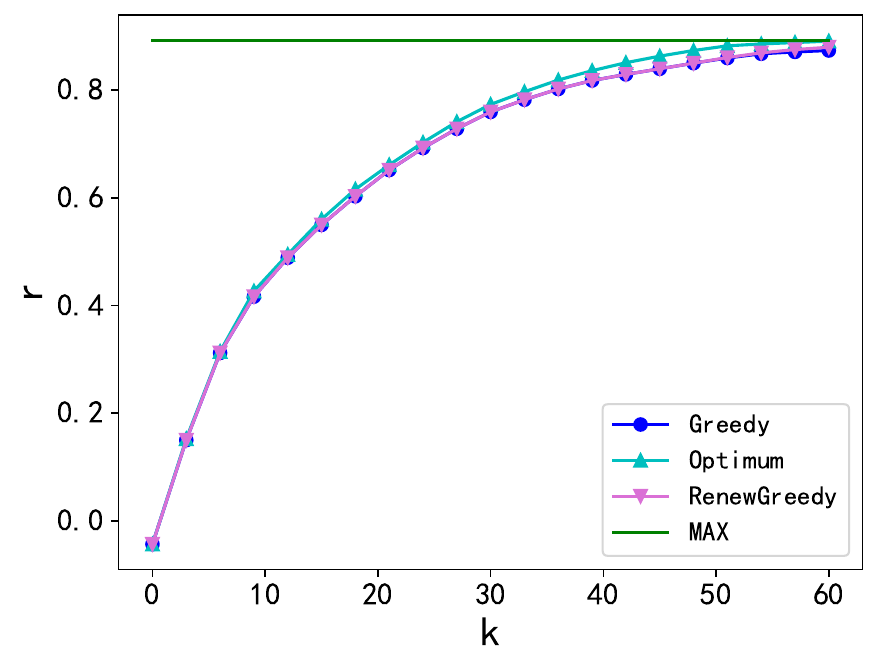}}
    \hspace{0.5cm}
	\subfloat[Dolphin]{
		\includegraphics[width=0.4\linewidth]{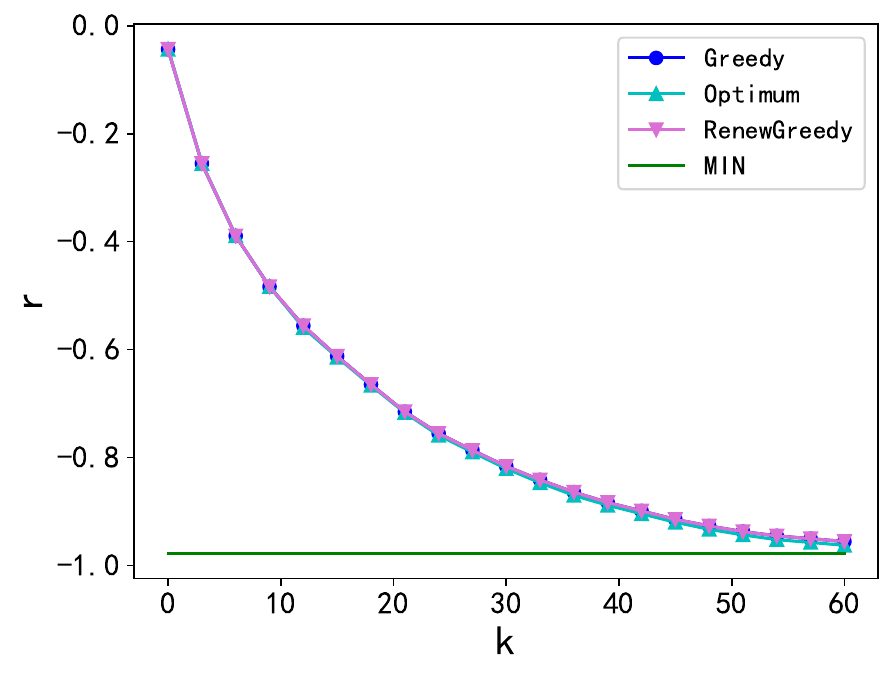}}
	\caption{The assortativity coefficient as a function of number $k$ of rewiring edges pairs for GRS, RenewGRS, the optimum solution, and the maximum or minimum assortativity coefficient on two networks.
 }
	\label{fig:22} 
\end{figure}
As shown in Figure \ref{fig:22}, as the value of $k$ increases, all three algorithms converge and approach a stable state. Furthermore, they closely approximate the maximum or minimum assortativity coefficient. 
When comparing the Optimum and GRS algorithms, we observe that our greedy algorithm performs exceptionally well, yielding results that are nearly identical to the optimal solution. 
Upon Comparing the GRS and RenewGRS algorithms, we find no significant differences, providing validation for some of our assumptions that newly generated edges are typically not selected for further rewiring. 
Additionally, the results suggest that performing a limited number of edge pair rewirings, typically less than half of the total number of edges, is adequate to achieve a network that is close to the maximum or minimum degree correlation.
\begin{figure*}[ht]
	\centering
	\subfloat[ER]{
		\includegraphics[scale=0.27]{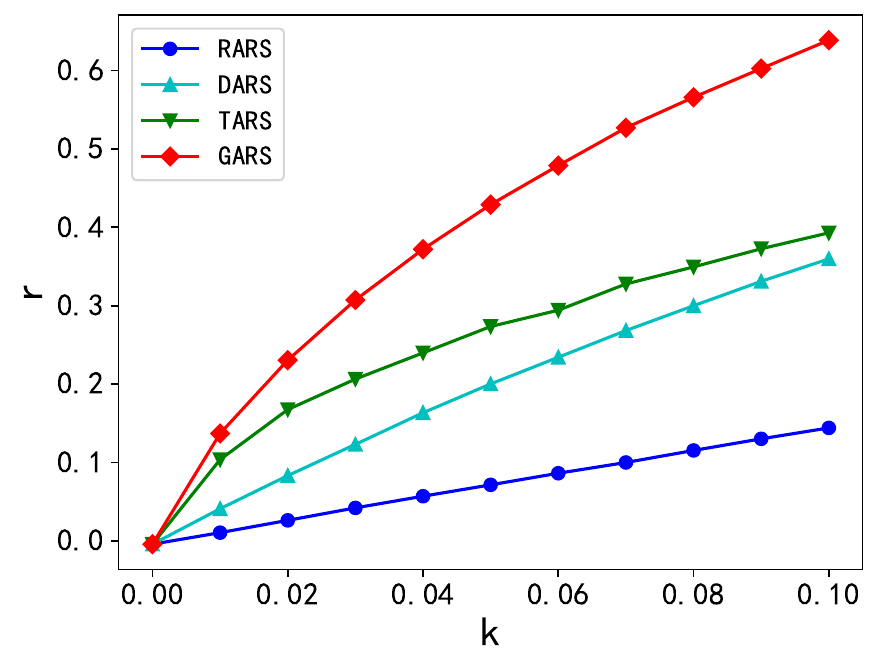}}
  \hspace{2cm}
	\subfloat[WS]{
		\includegraphics[scale=0.27]{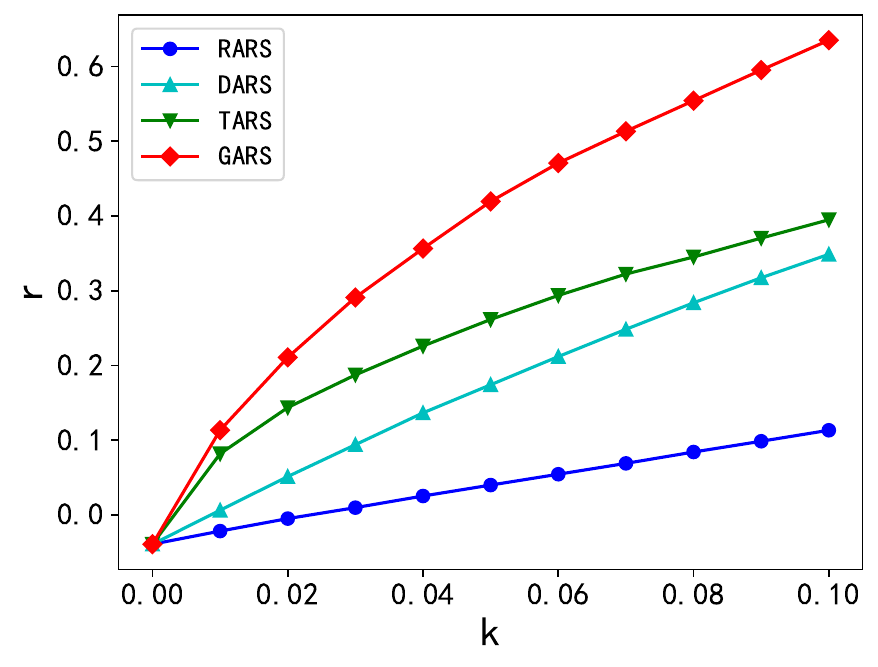}}
  \hspace{2cm}
        \subfloat[BA]{
		\includegraphics[scale=0.27]{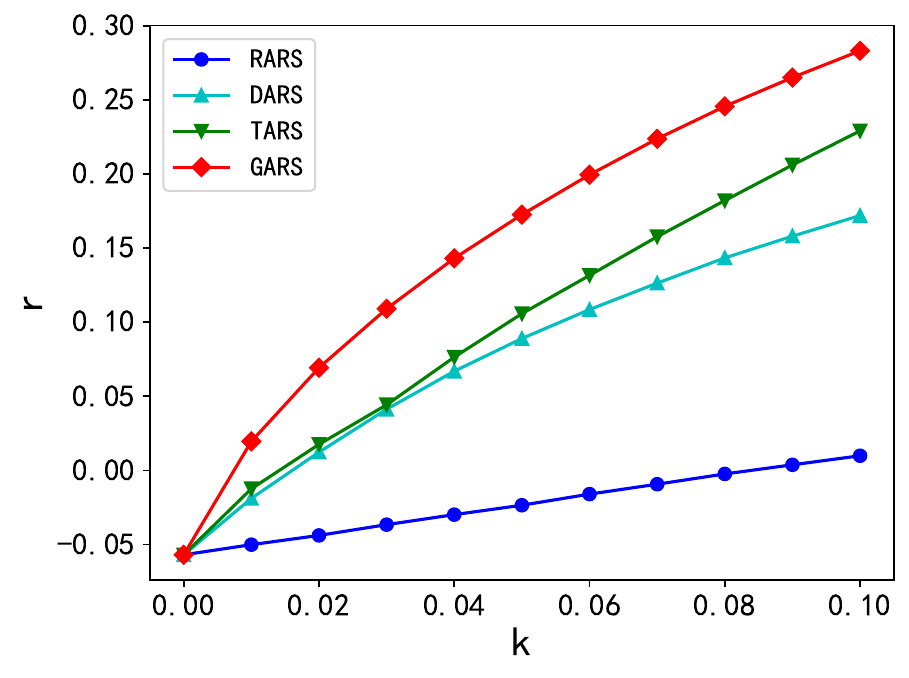}}
  \vspace{-0.3cm}
	\\
	\subfloat[Powergrid]{
		\includegraphics[scale=0.27]{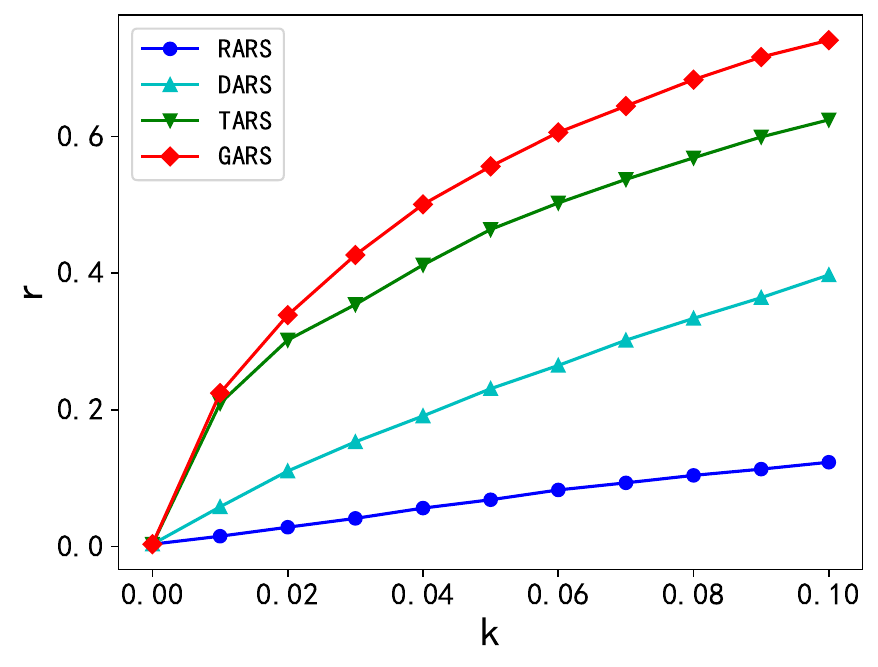}}
  \hspace{2cm}
	\subfloat[Netscience]{
		\includegraphics[scale=0.27]{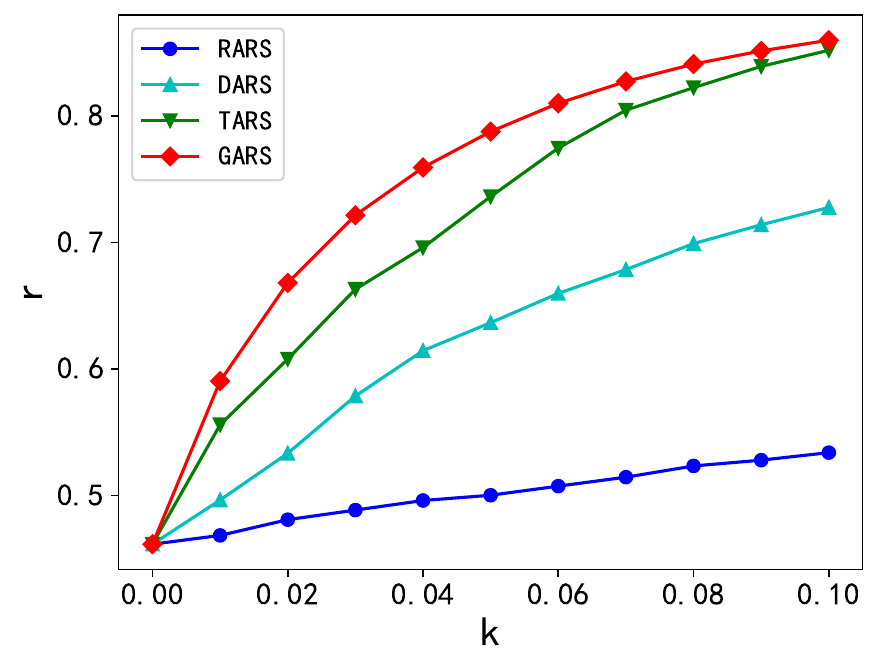}}
  \hspace{2cm}
        \subfloat[Metabolic]{
		\includegraphics[scale=0.27]{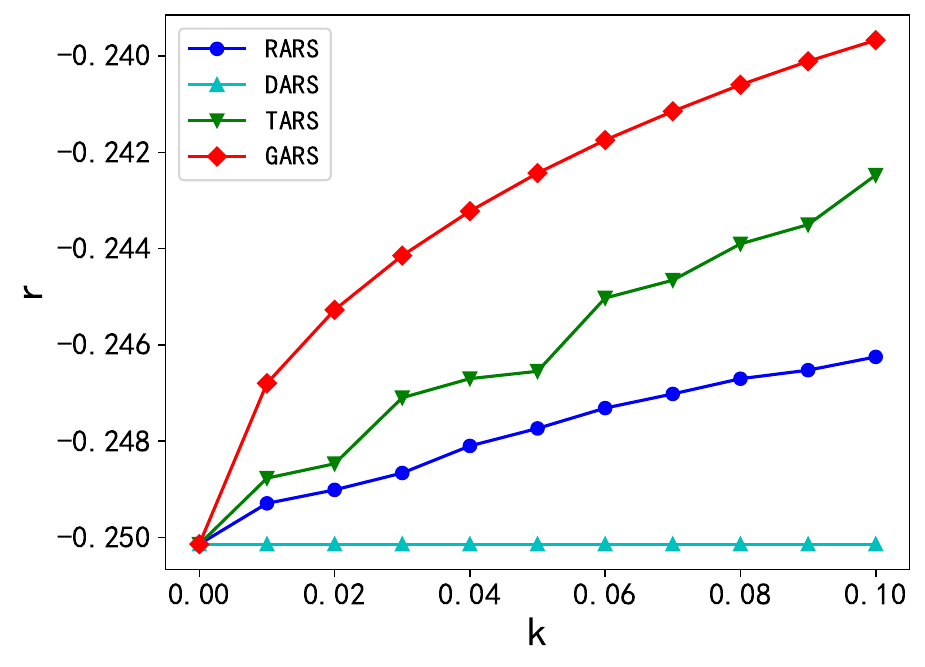}}
	\caption{Variations of assortativity coefficient in synthetic and real networks with respect to the percentage of rewired edges pairs under RARS, DARS, TARS, and GARS strategies. 
 }
	\label{fig:1} 
\end{figure*}
\begin{figure*}[!t]
	\centering
	\subfloat[ER]{
		\includegraphics[scale=0.27]{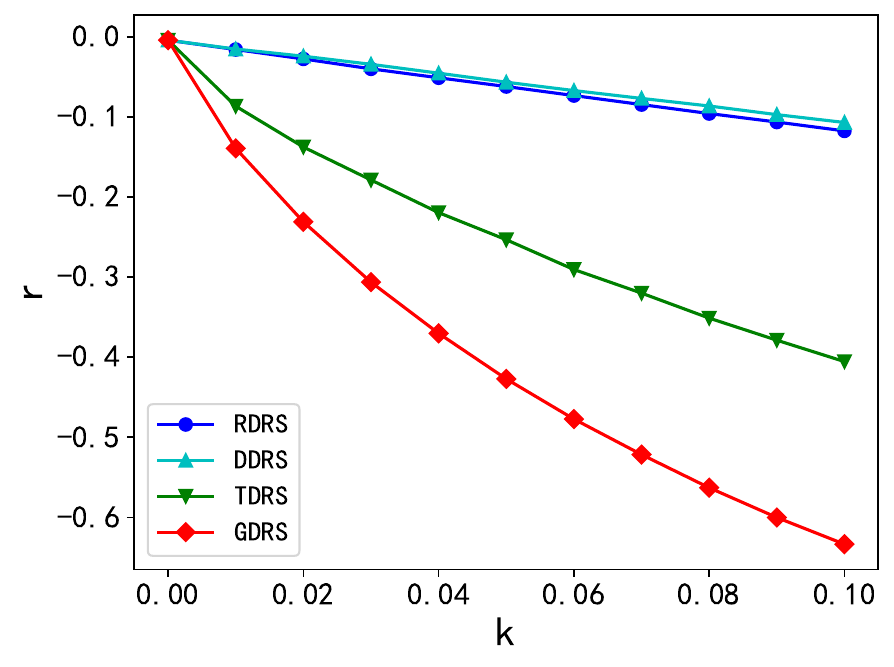}}
  \hspace{2cm}
	\subfloat[WS]{
		\includegraphics[scale=0.27]{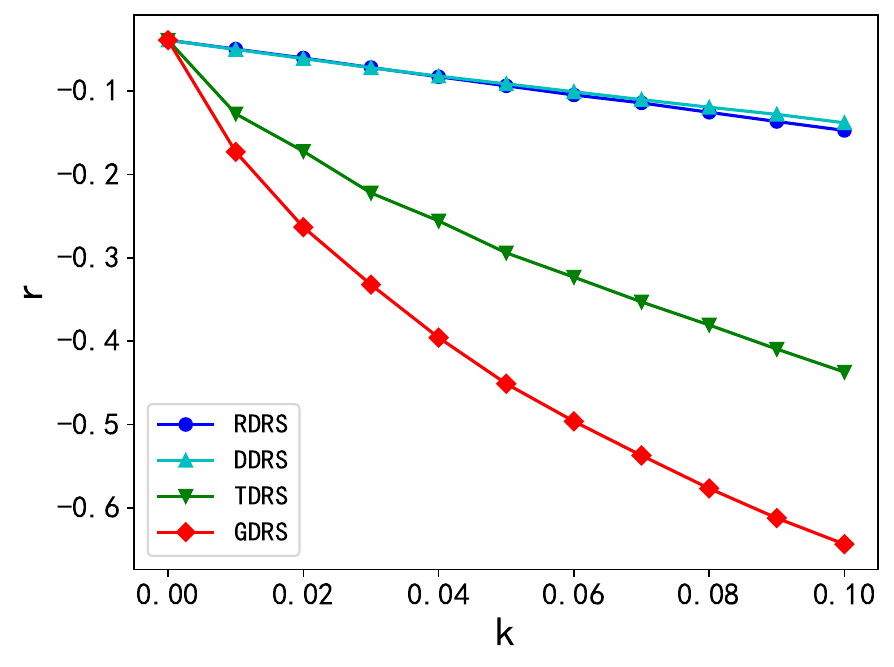}}
  \hspace{2cm}
        \subfloat[BA]{
		\includegraphics[scale=0.27]{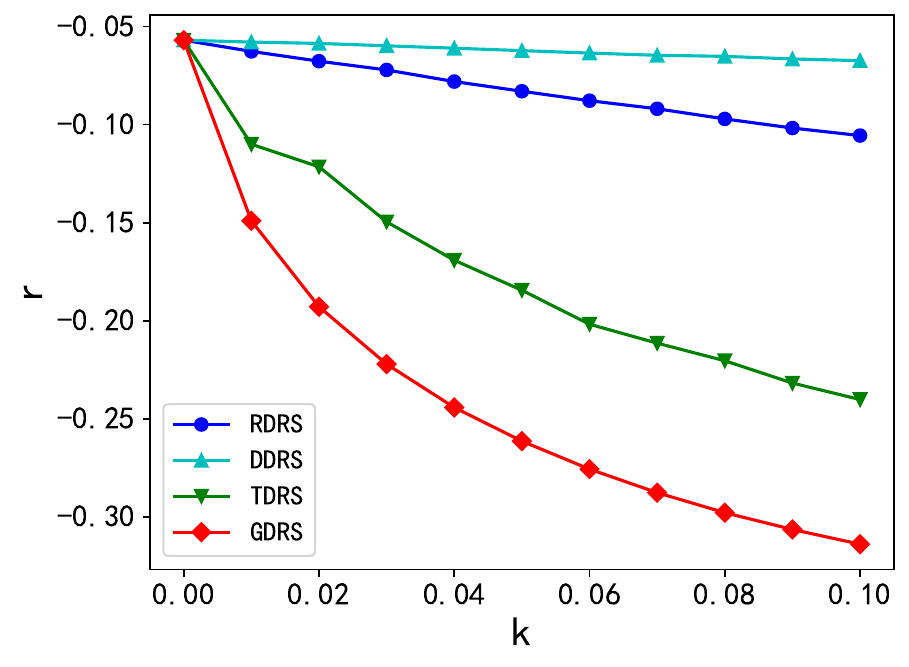} }
  \vspace{-0.3cm}
	\\
	\subfloat[Powergrid]{
		\includegraphics[scale=0.27]{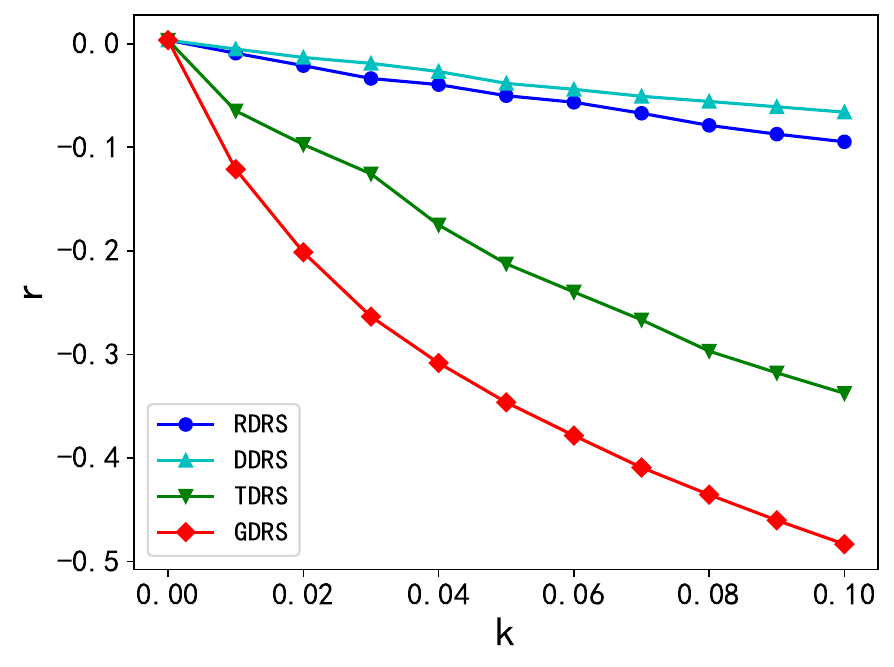}}
  \hspace{2cm}
	\subfloat[Netscience]{
		\includegraphics[scale=0.27]{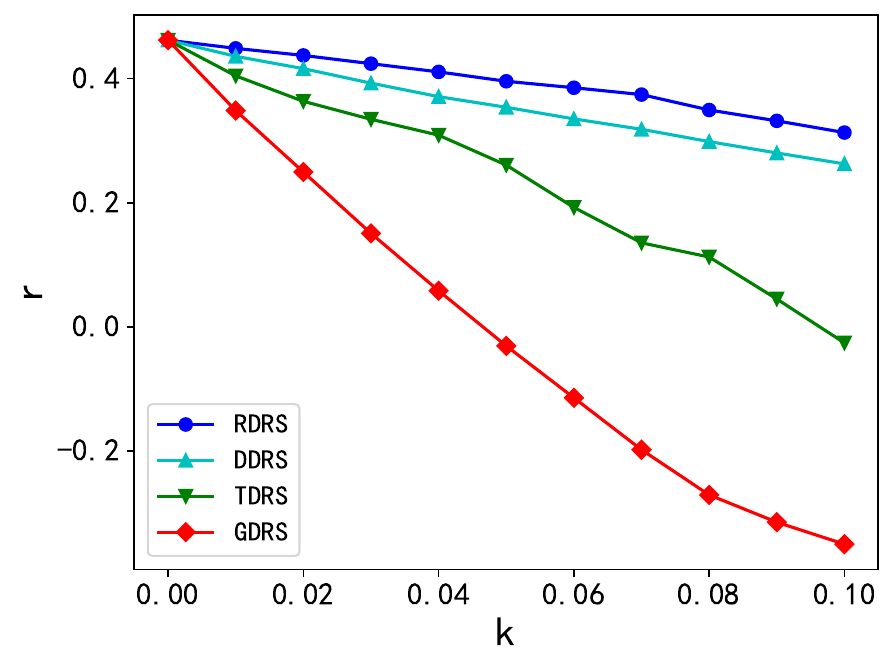}}
  \hspace{2cm}
        \subfloat[Metabolic]{
		\includegraphics[scale=0.27]{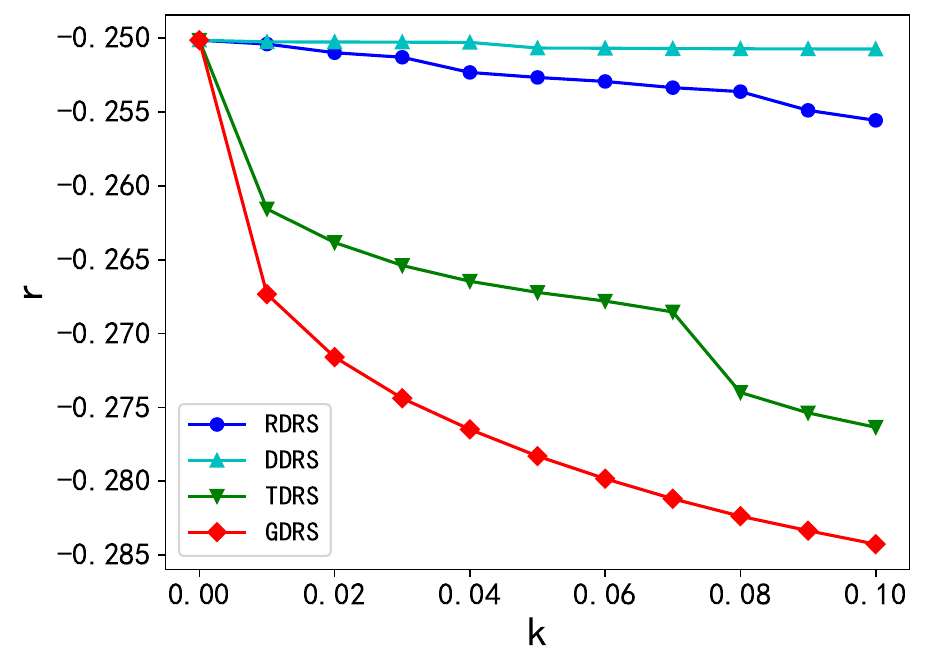}}
	\caption{Variations of assortativity coefficient in synthetic and real networks with respect to the percentage of rewired edges pairs under RDRS, DDRS, TDRS, and GDRS strategies. 
 }
	\label{fig:2} 
\end{figure*}
\subsection{The Comparison with Alternative Baselines}

To demonstrate the efficiency of our algorithm, we conducted a comparative analysis with TRS, RRS and DRS.

\textbf{Strategy 2 (Target Rewiring Strategy, TRS)}

Geng~\emph{et al.}\cite{geng2021global} proposed a target disassortative rewiring strategy (TDRS) that focuses on creating a disassortative network by prioritizing connections between high-degree nodes and low-degree nodes. Building on a similar strategy, we introduce a target assortative rewiring strategy(TARS) that prioritizes connections between high-degree nodes, thereby promoting assortativity in the network.


\textbf{Strategy 3 (Random Rewiring Strategy, RRS)}

Xulvi~\emph{et al.}~\cite{xulvi2005changing} proposed a random rewiring that preserves the degree distribution of a network. Two edges are randomly selected from the network. First, four nodes corresponding to the two edges are sorted by degree $k_1 \geq k_2 \geq k_3 \geq k_4$. If we use the random assortative rewiring strategy(RARS), we will connect the nodes corresponding to pairs $(k_1, k_2)$ and $(k_3, k_4)$. If we use the Random assortative rewiring strategy (RARS), we will connect the nodes corresponding to pairs $(k_1, k_4)$ and $(k_2, k_3)$.



\textbf{Strategy 4 (Degree-diff Rewiring Strategy, DRS)}

In assortative networks, the degree difference between node pairs is small, whereas, in disassortative networks, the degree difference tends to be large. On the basis of this observation, we propose a heuristic algorithm that considers the degree of difference between nodes. For each edge in the graph, calculate the degree difference between the two corresponding nodes, record the degree difference value of each edge as $diff_{e_{ij}}$, and sort the edges in descending order according to their degree differences as $diff$. If we use the Degree-diff assortative rewiring strategy(DARS), randomly select one pair of edges from the top 30\% of the $diff$ list and perform the rewiring process using the same method as RARS. If we use the Degree-diff disassortative rewiring strategy(DDRS), randomly select one pair of edges from the last 30\% of the $diff$ list and perform the rewiring process using the same method as RDRS.

\subsection{Results and Analysis}
To compare the effectiveness of the GRS with the aforementioned baseline methods, we carried out experiments on a subset of synthetic and real-world networks as shown in Table \ref{tab1}. The maximum attack budget in our experiments was set at $10\%$ of the total number of edges in the network. It should be noted that for random selection algorithms like RRS, we performed 100 iterations and averaged the results.

The results, as shown in Figures \ref{fig:1} and \ref{fig:2}, consistently demonstrate that GRS achieves the best performance among all algorithms. Each algorithm exhibits similar effects on ER networks and the corresponding effects on WS networks, indicating that although these networks possess different topological structures and degree distributions, these properties do not significantly alter their degree correlation. Moreover, compared to BA networks, ER and WS networks are more sensitive to rewiring. In the case of GRS, the changes of the assortativity coefficient in BA networks range roughly between $-0.3$ and $0.3$, while in the ER and WS networks, the changes range approximately between $-0.6$ and $0.6$.

In the Powergrid and Netscience networks, which initially exhibit assortative or neutral assortativity, GRS effectively modifies the network's assortativity coefficient. However, for the Metabolic network, which initially demonstrates disassortative due to its truncated structure, it is challenging to significantly change its assortativity through rewiring. In the case of GARS, the assortativity coefficient of the Metabolic network only improved by $0.011$. 

We also observed that the strategy based on the degree difference yielded similar results to the random strategy in terms of inducing network disassortativity. It is because edges with small degree differences only ensure similarity in the degrees of the two ends, without capturing the degree magnitudes. Consequently, when rewiring, it is possible to select pairs of edges with small degree differences but similar degrees, which do not significantly alter the network's assortativity coefficient.
\section{Conclusion}\label{thi}
We analyzed the factors that influence changes in the assortativity coefficient under degree-preserving conditions. Based on our assumptions, we formulate the problem of maximizing or minimizing the assortativity coefficient and verify its monotonic submodularity. From this we proposed a greedy rewiring strategy. Our experimental results provided strong evidence supporting the validity of our assumptions and demonstrated that our algorithm achieves results very close to the optimal solution. Our algorithm exhibited the best performance in both synthetic and real networks.

\bibliography{ref}
\bibliographystyle{IEEEtran}

\end{document}